\documentclass[12pt,a4paper]{article}

\usepackage[margin=1in]{geometry} 
\usepackage{amsfonts,amscd,amssymb, mathtools,mathrsfs, dsfont}
\usepackage[amsmath,amsthm,thmmarks]{ntheorem} 
\usepackage{graphicx,xypic,color,float} 
\usepackage{indentfirst}
\usepackage{lmodern}
\usepackage[colorlinks=true,citecolor=blue]{hyperref}

\newtheorem{thm}{Theorem}[section]
\newtheorem{prop}[thm]{Proposition}
\newtheorem{coro}[thm]{Corollary}
\newtheorem{lemma}[thm]{Lemma}

\newtheorem{remark}{Remark}[section]

\numberwithin{equation}{section} 

\renewcommand{\geq}{\geqslant}
\renewcommand{\leq}{\leqslant}

\newcommand{\citeprop}[1]{Proposition \ref{#1}}

\newcommand{\citelem}[1]{Lemma \ref{#1}}

\newcommand{\opfont}{\mathbf}

\newcommand{\BE}[2][]{{\opfont{E}}^{#1}\!\left[#2\right]}
\newcommand{\bp}{\opfont{P}}
\newcommand{\BP}[2][]{{\opfont{P}}^{#1}\!\left(#2\right)}
\newcommand{\BF}{\mathcal{F}}
\newcommand{\BL}{\mathcal{L}}

\newcommand{\BFF}{\mathds{F}}

\newcommand{\dd}{\operatorname{d}\!}
\newcommand{\dt}{\operatorname{d}\! t}

\newcommand{\dnu}{\operatorname{d}\! \nu}

\newcommand{\seta}{\mathcal{A}}

\newcommand{\sett}{\mathcal{T}}
\newcommand{\setc}{\mathcal{C}}

\newcommand{\setr}{\mathcal{R}}

\newcommand{\nn}{\nonumber}
\newcommand{\p}{\partial}

\begin{document}
\title{Optimal redeeming strategy of stock loans\\ under drift uncertainty}
\author{Zuo Quan Xu\thanks{Department of Applied Mathematics, The Hong Kong Polytechnic University, Kowloon, Hong Kong. This author acknowledges the financial support received from NSFC (No.~11471276), and Hong Kong GRF (No.~15204216 and No.~15202817). Email:~\texttt{maxu@polyu.edu.hk}. }\and Fahuai Yi\thanks{School of Finance, Guangdong University of Foreign Studies, Guangzhou, China. This author acknowledges the financial supports from NNSF of China (No.~11771158), NSF Guangdong Province of China (No.~2016A030313448, No.~2017A030313397). Email:~\texttt{fhyi@scnu.edu.cn}.}}

\date{\today}
\maketitle

\begin{abstract}
In practice, one must recognize the inevitable incompleteness of information while making decisions. In this paper, we consider the optimal redeeming problem of stock loans under a state of incomplete information presented by the uncertainty in the (bull or bear) trends of the underlying stock. This is called drift uncertainty. Due to the unavoidable need for the estimation of trends while making decisions, the related Hamilton-Jacobi-Bellman (HJB) equation is of a degenerate parabolic type. Hence, it is very hard to obtain its regularity using the standard approach, making the problem different from the existing optimal redeeming problems without drift uncertainty. We present a thorough and delicate probabilistic and functional analysis to obtain the regularity of the value function and the optimal redeeming strategies. The optimal redeeming strategies of stock loans appear significantly different in the bull and bear trends. \bigskip
\end{abstract}

\section{Introduction}
\noindent
The classical way to resolve optimal stopping problems (based on the dynamic programming principle) assumes that there is a unique known subjective prior distribution driving the underlying process. In the context of the classical Black-Scholes framework, in which the underlying asset process follows the geometric Brownian motion with known drift, the theory of financial optimal stopping problems has been well established (see, e.g., \cite{DJZZ09}, \cite{DTX12}, \cite{DX11}, \cite{DZ12}, \cite{G01a}, \cite{LZ13}, \cite{LTWW13}, \cite{SXZ08}, \cite{XZ07}, \cite{XZ13}, \cite{Xu16}, and \cite{ZJA12}). 
\par
In practice, on the other hand, while designing and controlling physical or organizational systems, one must recognize the inevitable incompleteness of information. For a wide range of applications in areas such as engineering, economics, and finance, the classical stochastic control theory, which is typically based on a single complete nominal model of the system, fails to provide strategies that yield satisfactory performance. Today, people are becoming increasingly aware of the importance of taking information incompleteness into account while dealing with stochastic control problems. In financial applications, as pointed out by Ekstr and Vaicenavicius \cite{EV15}, it is usually too strong to assume that the drift of the underlying asset is known. To obtain a reasonable and precise estimation of the drift, one needs a very long time series, which is rarely available, especially for an initial public offering stock for which the price history simply does not exist. To address this issue, various models have been proposed in the literature (see, e.g., \cite{G01b}, \cite{KX01}, \cite{XiongZhou2007}, \cite{S08}, \cite{SN08}, \cite{S09}, \cite{PZ10}, \cite{ZS13}, \cite{SZZ14}, and \cite{ZS14}). 
\par
Surprisingly, there have been only a few attempts to investigate financial optimal stopping problems under a state of incomplete information (\cite{DMV05}, \cite{EV15}, and \cite{V17}). In this paper, we study the optimal redeeming problem of stock loans under a state of incomplete information. A stock loan is a loan between a client (borrower) and a bank (lender), secured by a stock, which gives the borrower the right to redeem the stock at any time, on or before maturity, by repaying the lender the principal and a predetermined loan interest, or by surrendering the stock. Xia and Zhou \cite{XZ07} initiated the theoretical study of stock loan redemption (or equivalently, pricing) problem under the Black-Scholes framework. Through a probabilistic argument, they obtained a closed-form pricing formula for the standard stock loan for which the dividends are gained by the lender before redemption. Dai and Xu \cite{DX11} extended Xia and Zhou's work to general stock loans with different ways of dividend distributions through a PDE argument. Cai and Sun \cite{CS14} and Liang, Nie, and Zhao \cite{LNZ12} considered models with jumps. Zhang and Zhou \cite{ZZ00} and Prager and Zhang \cite{PZ10} studied models under regime-switching. 
\par
In regime-switching models, the current trend and volatility of the underlying stock are known to the borrower, although the trends may change anytime in the future. In contrast, in this paper, we assume that the borrower does not know the current trend of the stock and so she/he has to make decisions based on incomplete information. We choose to model the inherent uncertainty in the trends in the stock, called drift uncertainty, using a two-state random variable representing the bull and bear trends. The corresponding Hamilton-Jacobi-Bellman (HJB) equation turns out to be a degenerate parabolic one, which makes our problem far more challenging than the existing ones without drift uncertainty. The degeneracy is essentially due to the presence of the drift uncertainty in the model, and thus it cannot be removed by changing variables or using other standard ways in PDE. This unique feature leads to a failure in applying the PDE argument used in \cite{DX11} to tackle the present problem. In fact, the regularity of the value function is not good enough to let the HJB equation hold almost everywhere. Instead, it holds only in weak sense or in viscosity sense. In this paper, we present a thorough and delicate probabilistic and functional analysis to obtain the regularity of the value function and the optimal redeeming strategies. The optimal redeeming strategies of stock loans appear significantly different in the bull and bear trends. 
\par
The rest of this paper is organized as follows. We formulate the optimal redeeming problem of stock loans under drift uncertainty in Section 2. Some preliminary results on continuities of the value function are given in Section 3. We study the boundary cases in Section 4 and the general case in Section 5. A short conclusion is presented in Section 6.


\section{Problem formulation}
\subsection{The market and the underlying stock}
\noindent
We fix a filtered probability space $(\Omega, \BF,\BFF, \bp)$, which represents the financial market.
As usual, the filtration $\BFF=\{\BF_{t}\}_{t\geq 0}$ satisfies the usual conditions, and $\bp$ denotes the probability measure. In this probability space, there exists a standard one-dimensional Brownian motion $W$. The price process of the underlying stock is denoted as $S=(S_{t})_{t\geq 0}$ which evolves according to the stochastic differential equation (SDE):
\[\dd S_{t}=({\mu}-\delta) S_{t}\dt+\sigma S_{t}\dd W_{t},\]
where the dividend $\delta\geq 0$ and the volatility $\sigma>0$ are two known constants, while the return rate $\mu$ is unknown. 
By the self-similarity property of the Brownian motion, without loss of generality, we assume $\sigma=1$ throughout the paper.
\par
To model drift uncertainty, we assume $\mu$ is independent of the Brownian motion $W$, and $\mu-\delta$ may only take two possible values $a$ and $b$ that satisfy \[\Delta:=a-b>0.\] The stock is said to be in its \emph{bull trend} when $\mu-\delta=a$, and in its \emph{bear trend} when $\mu-\delta=b$.

\begin{remark}
If $\Delta=0$, then the drift uncertainty disappears and our financial market reduces to the classical Black-Scholes one. 
\end{remark}

\subsection{The stock loan and its optimal redeeming problem }
\noindent
A stock loan is a loan secured by a stock, which gives the borrower the right to redeem the stock at any time before or upon maturity. In this paper, we will study only the standard stock loan in which dividends are gained by the lender before redemption. For such a stock loan, when the borrower redeems the stock at time $t$, she/he has to repay the amount of $Ke^{\gamma t}$ to the lender and get back the stock (but the stock dividend gained up to time $t$ will be left with the lender). Here $K>0$ represents the loan principle (also called loan to value), and $\gamma$ represents the loan rate. Such a stock loan has been considered under different frameworks (\cite{DX11}, \cite{XZ07}, and \cite{ZZ00}). In practice, the loan rate $\gamma$ should be higher than the risk-free interest rate $r$: \[\gamma >r,\] which is henceforth assumed for the rest of this paper.
\par
This paper studies the borrower's optimal redeeming problem: that is, determining the optimal stopping time to achieve
\begin{align}\label{main000}
\sup_{\tau\in\sett_{[t,T]}}\;\BE{e^{-r(\tau-t)}\left(S_{\tau}-Ke^{\gamma \tau}\right)^{+}\;\big|\; \BF^{S}_{t}},
\end{align}
where $\sett_{[t,T]}$ denotes the set of all $\BFF^{S}$-stopping times valued in $[t,T]$. Here, $\BFF^{S}=\{\BF_{t}^{S}\}_{t\geq 0}$ is the natural filtration generated by the stock price process, which is observable for the borrower. 
If $\gamma=0$, the problem \eqref{main000} reduces to the optimal redeeming problem of the vanilla American call option under drift uncertainty (see more discussions in \cite{XZ07}). In our model, the trend of the underlying stock is unknown; hence, this is not a standard optimal stopping problem.
\par
To determine the optimal redeeming strategy, the borrower has to estimate the current trend of the stock first. For this, introduce the \emph{a posteriori probability process} $\pi=(\pi_{t})_{t\geq 0}$ defined as
\[\pi_{t}:=\BP{\mu-\delta=a\:\big|\: \BF_{t}^{S}}.\]
Roughly speaking, it estimates the probability that the stock is in its bull trend at time $t$. We assume $0<\pi_{0}<1$, otherwise the drift uncertainty disappears. 
\par
Sometimes, it is more convenient to consider the log-price process $L:=(\log S_{t})_{t\geq 0}$ which by It\^{o}'s lemma follows
\[\dd L_{t}=(\mu-\delta-\tfrac{1}{2})\dt+ \dd W_{t}.\] 
We notice that $\BFF^{S}$ is the same as $\BFF^{L}$, the filtration generated by $L$. According to the innovation representation (see, e.g., Proposition 2.30 in \cite{BC09}), the process
\[\overline{W}_{t}=L_{t}-\int_{0}^{t}\BE{\mu-\delta-\tfrac{1}{2}\:\big|\: \BF_{\nu}^{L}}\dnu=L_{t}-\int_{0}^{t}\left(\Delta\pi_{\nu}+b-\tfrac{1}{2}\right)\dnu\]
is a Brownian motion under the (observable) filtration $\BFF^{L}$. It then follows
\begin{align}\label{pro:L}
\dd L_{t}=\left(\Delta\pi_{t}+b-\tfrac{1}{2}\right)\dt+\dd \overline{W}_{t},
\end{align}
and applying It\^{o}'s lemma yields 
\begin{align}\label{pro:S}
\dd S_{t}=\left(\Delta\pi_{t}+b\right)S_{t}\dt+S_{t}\dd \overline{W}_{t}.
\end{align}
We notice that $b\leq \Delta\pi_{t}+b\leq a$, and this will be used frequently in the subsequent analysis without claim. 
An application of the general Bayes' formula (see, e.g., Chapter 7.9 in \cite{LS2001a}) and It\^{o}'s lemma give
\begin{align}\label{pro:pi}
\dd \pi_{t}=\Delta\pi_{t}(1-\pi_{t})\dd \overline{W}_{t}.
\end{align} 
Therefore, solving the problem \eqref{main000}, regarded as an optimal stopping problem of the Markovian processes \eqref{pro:S} and \eqref{pro:pi}, reduces to determining 
\begin{align}\label{main001}
V(s,\pi,t):=\sup_{\tau\in\sett_{[t,T]}}\;\BE{e^{-r(\tau-t)}\left(S_{\tau}-Ke^{\gamma \tau}\right)^{+}\big| S_{t}=s, \pi_{t}=\pi}
\end{align}
for $(s,\pi,t)$ in the domain \[\seta :=(0,+\infty)\times(0,1)\times[0,T).\] 
This problem is an optimal stopping problem of an observable Markov process $(S_{t}, \pi_{t})_{t\geq 0}$, hence the dynamic programming principle may be applied and one would expect to use the variational inequality techniques in PDE to tackle the corresponding HJB equation. 
Dai and Xu \cite{DX11} used this method to study the optimal redeeming problem of stock loans in case without drift uncertainty. Without drift uncertainty, the HJB equation in \cite{DX11} is uniformly parabolic. In contrast, in our problem, it is degenerate parabolic, for which it is difficult to obtain regularity from the PDE perspective,\footnote{See \cite{CYW11} for an example of the study of a degenerate parabolic equation.} making our problem more challenging than before. In this paper, we present a thorough and delicate probabilistic and functional analysis to the value function to obtain its regularity and optimal redeeming strategies. 
\par
Since the obstacle in the problem \eqref{main001} is time variant, it is convenient to introduce the discounted stock process 
$X_t=e^{-\gamma t}S_t$, which by It\^{o}'s lemma follows
\begin{align}\label{pro:X}
\dd X_t=\big(\Delta\pi_t+b-\gamma\big)X_t\dt+X_t\dd\overline{W}_t.
\end{align}
If we define 
\begin{align}\label{main00}
u(x,\pi,t):=\sup_{\tau\in\sett_{[t,T]}}\;\BE{e^{(\gamma-r)(\tau-t)}\left(X_{\tau}-K\right)^{+}\big| X_{t}=x, \pi_{t}=\pi}.
\end{align}
Then, one can easily see that 
\begin{align}\label{V=u}
V(s,\pi,t)=e^{\gamma t}u(e^{-\gamma t}s,\pi, t),\quad (s,\pi,t)\in\seta.
\end{align} 
Thus, solving the problem \eqref{main001} reduces to solving \eqref{main00}. From now on, we call $u(x,\pi,t)$ defined in \eqref{main00} the value function, and pay attention to it rather than to $V(s,\pi,t)$. 
\par
Define the continuation region
\begin{align}\label{C}
\setc&:=\{(x,\pi,t)\in\seta\mid u(x,\pi,t)>(x-K)^+\},
\end{align}
and the redeeming region
\begin{align}\label{R}
\setr&:=\{(x,\pi,t)\in\seta\mid u(x,\pi,t)=(x-K)^+\}.
\end{align}
By the general theory of optimal stopping, the optimal redeeming strategy for the problem \eqref{main00} as well as \eqref{main001} is given by the hitting time of the redeeming region $\setr$: that is, 
\begin{align*} 
\tau^{*}&=\inf\{t\in [0,T)\mid (X_{t},\pi_{t},t)\in\setr\}\wedge T\\
&=\inf\{t\in [0,T)\mid (e^{-\gamma t}S_t,\pi_{t},t)\in\setr\}\wedge T.
\end{align*}
It is not optimal to redeem the stock in the continuation region $\setc$. 
\par
The main purpose of this paper is then reduced to studying the value function $u(x,\pi,t)$ and determining the redeeming region $\setr$ and the continuation region $\setc$.
\par
In the problem \eqref{main00}, it is easily seen by choosing $\tau=T$ that $u>0$ on $\seta$. Consequently,
\begin{align}\label{regions}
\setc \supseteq \{(x,\pi,t)\in\seta \mid x\leq K\},\quad
\setr \subseteq \{(x,\pi,t)\in\seta \mid x> K\}.
\end{align} 
Therefore, the continuation region $\setc$ is always non-empty. In contrast, the redeeming region $\setr$, depending on the parameters as shown below, can be either empty or non-empty.

\section{Preliminaries and continuities}
\noindent
In this section, we study the continuities of the value function defined in \eqref{main00}, mainly using probabilistic analysis. 
\par

\subsection{Preliminaries}
\noindent
We will use the following elementary inequalities frequently in the subsequent analysis.
\begin{lemma}\label{lem:ineqs}
For any real-valued functions $f$, $g$, and real numbers $x$, $y$, we have
\begin{enumerate}
\item $|\sup f-\sup g|\leq \sup |f-g|$;
\item $|x^{+}-y^{+}|\leq |x-y|$; in particular, $|(x-K)^{+}-(y-K)^{+}|\leq |x-y|$;
\item $|e^{x}-e^{y}|\leq (e^{x}+e^{y})|x-y|$. 
\end{enumerate}
\end{lemma}
\begin{proof}
The last inequality follows from the mean value theorem. The others are easy to check. 
\end{proof}

The following result gives an upper bound for the value function. As a consequence, the problem \eqref{main00} is meaningful. 
\begin{lemma}
We have
\begin{align*}
u(x,\pi,t)\leq xe^{(a-r)^{+}(T-t)},\quad (x, \pi, t)\in\seta.
\end{align*}
\end{lemma}
\begin{proof}
For any $\tau\in\sett_{[t,T]}$, applying It\^{o}'s lemma to \eqref{pro:X} gives
\[ X_{\tau}= X_{t}e^{\int_{t}^{\tau}(\Delta\pi_{\nu}+b-\gamma-\frac{1}{2})\dnu+\overline{W}_{\tau}-\overline{W}_{t} } 
\leq X_{t}e^{(a-\gamma-\frac{1}{2})(\tau-t)+\overline{W}_{\tau}-\overline{W}_{t} },\]
and thus
\begin{align*}
u(x,\pi,t)&=\sup_{\tau\in\sett_{[t,T]}}\;\BE{e^{(\gamma-r)(\tau-t)}\left(X_{\tau}-K\right)^{+}\big| X_{t}=x, \pi_{t}=\pi }\\
&\leq\sup_{\tau\in\sett_{[t,T]}}\;\BE{e^{(\gamma-r)(\tau-t)} X_{\tau}\big| X_{t}=x, \pi_{t}=\pi }\\
&\leq\sup_{\tau\in\sett_{[t,T]}}\;\BE{X_{t}e^{(a-r-\frac{1}{2})(\tau-t)+\overline{W}_{\tau}-\overline{W}_{t} }\big| X_{t}=x, \pi_{t}=\pi }\\
&=\sup_{\tau\in\sett_{[t,T]}}\;\BE{xe^{(a-r)(\tau-t)}e^{-\frac{1}{2}(\tau-t)+\overline{W}_{\tau}-\overline{W}_{t} }\big| \pi_{t}=\pi }\\
&\leq \sup_{\tau\in\sett_{[t,T]}}\;\BE{xe^{(a-r)^{+}(T-t)}e^{-\frac{1}{2}(\tau-t)+\overline{W}_{\tau}-\overline{W}_{t} }\big| \pi_{t}=\pi }\\
&=xe^{(a-r)^{+}(T-t)}.
\end{align*}
\end{proof}
\par
The next result states the monotonicity of the value function.
\begin{lemma}\label{umono}
The value function $u(x,\pi,t)$ on $\seta$ is non-decreasing in $\pi$, non-increasing in $t$, and non-decreasing and convex in $x$. 
\end{lemma}
\begin{proof}
Both $X_{t}$ and $\pi_{t}$ are stationary processes; so we have
\begin{align}\label{umonoexp}
u(x,\pi,t)&=\sup_{\tau\in\sett_{[t,T]}}\;\BE{e^{(\gamma-r)(\tau-t)}\left(X_{\tau}-K\right)^{+}\big| X_{t}=x, \pi_{t}=\pi }\nn\\
&=\sup_{\tau\in\sett_{[t,T]}}\;\BE{e^{(\gamma-r)(\tau-t)}\left(X_{t}e^{\int_{t}^{\tau}(\Delta\pi_{\nu}+b-\gamma-\frac{1}{2})\dnu+\overline{W}_{\tau}-\overline{W}_{t} }-K\right)^{+}\bigg| X_{t}=x, \pi_{t}=\pi }\nn\\
&=\sup_{\tau\in\sett_{[t,T]}}\;\BE{e^{(\gamma-r)(\tau-t)}\left(xe^{\int_{t}^{\tau}(\Delta\pi_{\nu}+b-\gamma-\frac{1}{2})\dnu+\overline{W}_{\tau}-\overline{W}_{t} }-K\right)^{+}\bigg| \pi_{t}=\pi }\nn\\
&=\sup_{\tau\in\sett_{[0,T-t]}}\;\BE{e^{(\gamma-r)\tau}\left(xe^{\int_{0}^{\tau}(\Delta\pi_{\nu}+b-\gamma-\frac{1}{2})\dnu+\overline{W}_{\tau} }-K\right)^{+}\bigg| \pi_{0}=\pi }.
\end{align} 
It follows that $u(x,\pi,t)$ is convex and non-decreasing in $x$, and non-increasing in $t$. 
\par
Now, let us show that $u$ is non-decreasing in $\pi$.
Without loss of generality, we may assume $t=0$. It\^{o}'s lemma yields
\begin{align*}
\dd \;\log\left(\frac{\pi_{t}}{1-\pi_{t}}\right)&=\dd \;\log \pi_{t}-\dd \;\log(1-\pi_{t})\\
&=\Delta (1-\pi_{t})\dd \overline{W}_{t}-\tfrac{1}{2}\Delta^{2}(1-\pi_{t})^{2}\dt+\Delta\pi_{t}\dd \overline{W}_{t}+\tfrac{1}{2}\Delta^{2} \pi_{t}^{2}\dt\\
&=\Delta^{2}\left(\pi_{t}-\tfrac{1}{2}\right)\dt+\Delta\dd \overline{W}_{t},
\end{align*}
so that
\begin{align*}
\log\left(\frac{\pi_{t}}{1-\pi_{t}}\right)=\log\left(\frac{\pi_{0}}{1-\pi_{0}}\right)+\Delta^{2}\int_{0}^{t} \left(\pi_{\nu}-\tfrac{1}{2}\right)\dnu+\Delta\overline{W}_{t}.
\end{align*}
If $\pi_{t}$ and $\pi'_{t}$ are the solutions of \eqref{pro:pi} with initial values $\pi_{0} >\pi'_{0}$ (both in $(0,1)$), then $\tau=\inf\{t\geq0\mid \pi_{t}\leq\pi'_{t}\}>0$. If $\tau(\omega)$ is finite, then, by continuity, we have $\pi_{\tau}(\omega)=\pi'_{\tau}(\omega)\in(0,1)$ and thus on this sample path $\omega$
\begin{multline*}
\qquad 0=\log\left(\frac{\pi_{\tau}}{1-\pi_{\tau}}\right)-\log\left(\frac{\pi'_{\tau}}{1-\pi'_{\tau}}\right) \\
=\log\left(\frac{\pi_{0}}{1-\pi_{0}}\right)-\log\left(\frac{\pi'_{0}}{1-\pi'_{0}}\right)+\Delta^{2}\int_{0}^{\tau} (\pi_{\nu}-\pi'_{\nu})\dnu>0,\qquad
\end{multline*}
a contradiction. Thus, $\tau(\omega)=+\infty$ almost surely, and consequently $\pi_{t}(\omega)>\pi'_{t}(\omega)$ for all $t$ almost surely. Therefore, $\pi_{t}$ is non-decreasing $\omega$-wisely with respect to the initial value $\pi_{0}$, and thus, from the last expression in \eqref{umonoexp} we see that $u$ is non-decreasing in $\pi$. 
\end{proof}
\begin{lemma}\label{lem:uxbound}
We have, uniformly, for all $ 0\leq y\leq x<\infty$, $\pi\in [0,1]$, and $ t\in [0,T]$, that
\begin{align}\label{uxbound1}
0\leq u(x,\pi,t)-u(y,\pi,t)\leq (x-y)e^{(a-r)^{+}(T-t)}.
\end{align}
Especially, if $r\geq a$, then 
\begin{align}\label{uxbound2}
0\leq u(x,\pi,t)-u(y,\pi,t)\leq x-y.
\end{align}
\end{lemma}
\begin{proof}
Applying \citelem{lem:ineqs} to the last expression in \eqref{umonoexp}, we have
\begin{align*}
|u(x,\pi,t)-u(y,\pi,t)|&\leq |x-y|\sup_{\tau\in\sett_{[0,T-t]}}\;\BE{e^{(\gamma-r)\tau}e^{\int_{0}^{\tau}(\Delta\pi_{\nu}+b-\gamma-\frac{1}{2})\dnu+\overline{W}_{\tau} }\bigg| \pi_{0}=\pi }\\
&\leq |x-y|\sup_{\tau\in\sett_{[0,T-t]}}\;\BE{e^{(a-r)\tau}e^{-\frac{1}{2}\tau+\overline{W}_{\tau} }\bigg| \pi_{0}=\pi }\\
&\leq |x-y| e^{(a-r)^{+}(T-t)}.
\end{align*}
\end{proof}

\begin{lemma}\label{lem:upibound}
We have, uniformly for all $x\in (0,\infty)$, $0\leq \varpi\leq \pi\leq 1$ and $ t\in [0,T]$, that
\begin{align}\label{upibound1}
0\leq u(x,\pi,t)- u(x,\varpi,t)\leq C\Delta x(T-t)(\pi-\varpi), 
\end{align} 
for some constant $C>0$. 
\end{lemma}
\begin{proof}
Let $\pi_{t}$ and $\pi'_{t}$ be the solutions of \eqref{pro:pi} with initial values $\pi \geq \varpi$, respectively.
Denote 
\begin{align*}
Y_{t}=\int_{0}^{t}(\Delta\pi_{\nu}+b-\gamma-\tfrac{1}{2})\dnu+\overline{W}_{t}\leq (a-\gamma+\tfrac{1}{2})^{+}t+\overline{W}_{t}-t, 
\end{align*}
and
\begin{align*}
Y'_{t}=\int_{0}^{t}(\Delta\pi'_{\nu}+b-\gamma-\tfrac{1}{2})\dnu+\overline{W}_{t}\leq (a-\gamma+\tfrac{1}{2})^{+}t+\overline{W}_{t}-t.
\end{align*}
Then 
\begin{align*}
\sup_{s\leq T}|Y_{s}-Y'_{s}|\leq \Delta T\sup_{s\leq T} |\pi_{s}-\pi'_{s}|.
\end{align*}
Similar to what was seen before, applying \citelem{lem:ineqs} to the last expression in \eqref{umonoexp} and using Cauthy's inequality, we have
\begin{align*}
|u(x,\pi,t)-u(x,\varpi,t)|&\leq x\sup_{\tau\in\sett_{[0,T-t]}}\;\BE{e^{(\gamma-r)\tau} \left|e^{Y_{\tau}}-e^{Y'_{\tau}}\right|\bigg| \pi_{0}=\pi,\pi'_{0}=\varpi }\\
&\leq xe^{(\gamma-r)(T-t)}\sup_{\tau\in\sett_{[0,T-t]}}\;\BE{ \left|e^{Y_{\tau}}-e^{Y'_{\tau}}\right|\bigg| \pi_{0}=\pi,\pi'_{0}=\varpi }\\
&\leq xe^{(\gamma-r)(T-t)}\sup_{\tau\in\sett_{[0,T-t]}}\;\BE{\left(e^{Y_{\tau}}+e^{Y'_{\tau}}\right)|Y_{\tau}-Y'_{\tau}|\big| \pi_{0}=\pi,\pi'_{0}=\varpi }\\
&\leq xe^{(\gamma-r)(T-t)}\sup_{\tau\in\sett_{[0,T-t]}}\;\BE{2e^{(a-\gamma+\tfrac{1}{2})^{+}(T-t)+\overline{W}_{\tau}-\tau}|Y_{\tau}-Y'_{\tau}|\big| \pi_{0}=\pi,\pi'_{0}=\varpi }\\
&\leq 2xe^{(\gamma-r+(a-\gamma+\tfrac{1}{2})^{+})(T-t)}\sup_{\tau\in\sett_{[0,T-t]}}\;\left(\BE{e^{2\overline{W}_{\tau}-2\tau}\big| \pi_{0}=\pi,\pi'_{0}=\varpi }\right)^{1/2}\\
&\quad\;\times\sup_{\tau\in\sett_{[0,T-t]}}\;\left(\BE{(Y_{\tau}-Y'_{\tau})^{2}\big| \pi_{0}=\pi,\pi'_{0}=\varpi }\right)^{1/2}\\
&\leq Cx \left( \Delta^{2}(T-t)^{2}\BE{\sup_{s\leq T-t}(\pi_{s}-\pi'_{s})^{2}\big| \pi_{0}=\pi,\pi'_{0}=\varpi }\right)^{1/2}\\
&\leq Cx \Delta (T-t)(\pi-\varpi).
\end{align*}
where the last inequality is due to (1.19), p.25, in \cite{P09}. 
Here, and hereafter, the implied constant $C$ may vary in each appearance. 
\end{proof}

\begin{lemma}\label{lem:upibound}
We have, uniformly for $x\in (0,\infty)$, $\pi\in[0,1]$ and $0\leq s\leq t\leq T$, that
\begin{align}\label{utbound1}
0\leq u(x,\pi,s)-u(x,\pi,t)\leq Cx(t-s)^{\frac{1}{2}},
\end{align} 
for some constant $C>0$.
\end{lemma}
\begin{proof}
Denote 
\begin{align*}
Y_{t}=\int_{0}^{t}(\Delta\pi_{\nu}+b-r-\tfrac{1}{2})\dnu+\overline{W}_{t}\leq (a-r-\tfrac{1}{2})t+\overline{W}_{t}.
\end{align*}
For any $\tau\in\sett_{[0,T-s]}$, set $\tau'=\tau\wedge (T-t)\in\sett_{[0,T-t]}$. Then $0\leq \tau-\tau'\leq t-s$ and 
\begin{align*} 
e^{(\gamma-r)\tau}(X_{\tau}-K)^{+}=( e^{(\gamma-r)\tau}X_{\tau}-K e^{(\gamma-r)\tau})^{+}\leq (X_{0}e^{Y_{\tau}}-K e^{(\gamma-r)\tau'})^{+}.
\end{align*} 
Using \citelem{lem:ineqs}, we have 
\begin{align*} 
0&\leq u(x,\pi,s)-u(x,\pi,t)\\
&\leq \sup_{\tau\in\sett_{[0,T-s]}}\;\BE{e^{(\gamma-r)\tau}(X_{\tau}-K)^{+}-e^{(\gamma-r)\tau'}(X_{\tau'}-K)^{+}\big| X_{0}=x,\pi_{0}=\pi }\\
&\leq\sup_{\tau\in\sett_{[0,T-s]}}\;\BE{(xe^{Y_{\tau}}-Ke^{(\gamma-r)\tau'})^{+}-(xe^{Y_{\tau'}}-Ke^{(\gamma-r)\tau'})^{+}\big| \pi_{0}=\pi }\\
&\leq x\sup_{\tau\in\sett_{[0,T-s]}}\;\BE{\big|e^{Y_{\tau}}-e^{Y_{\tau'}}\big|\Big|\pi_{0}=\pi }\\
&\leq x\sup_{\tau\in\sett_{[0,T-s]}}\;\BE{\big(e^{Y_{\tau}}+e^{Y_{\tau'}}\big)\big|Y_{\tau}-Y_{\tau'}\big|\Big|\pi_{0}=\pi }\\
&\leq x\sup_{\tau\in\sett_{[0,T-s]}}\;\left(\BE{\big(e^{Y_{\tau}}+e^{Y_{\tau'}}\big)^{2}\Big|\pi_{0}=\pi }\right)^{1/2}
\sup_{\tau\in\sett_{[0,T-s]}}\;\left(\BE{\big(Y_{\tau}-Y_{\tau'}\big)^{2}\Big|\pi_{0}=\pi }\right)^{1/2}\\
&\leq Cx \left(\sup_{\tau\in\sett_{[0,T-s]}}\;\BE{\left(\int_{\tau'}^{\tau}(\Delta\pi_{\nu}+b-r-\tfrac{1}{2})\dnu+\overline{W}_{\tau}-\overline{W}_{\tau'}\right)^{2}\Big|\pi_{0}=\pi }\right)^{1/2}\\
&\leq Cx \left(\sup_{\tau\in\sett_{[0,T-s]}}\;\BE{\left(C(t-s)+|\overline{W}_{\tau}-\overline{W}_{\tau'}|\right)^{2} }\right)^{1/2}\\ 
&\leq Cx \left((t-s)^{2}+\sup_{\tau\in\sett_{[0,T-s]}}\;\BE{\left(\overline{W}_{\tau}-\overline{W}_{\tau'}\right)^{2}-(\tau-\tau') }+\sup_{\tau\in\sett_{[0,T-s]}}\;\BE{\tau-\tau'}\right)^{1/2}\\
&\leq Cx(t-s)^{1/2}.
\end{align*} 
\end{proof}
Since the monotonic function is differentiable almost everywhere, the above lemmas lead to
\begin{coro}
The value function $u(x, \pi, t)$ is continuous in $\overline{\seta}:=[0,\infty)\times[0,1]\times[0,T]$. The partial derivatives $u_{x}$, $u_{\pi}$ and $u_{t}$ exist almost everywhere in $\seta$; and uniformly in $\overline{\seta}$, we have
\begin{gather*}
0\leq u_{x}\leq C,\\
0\leq u_{\pi}\leq Cx,\\
u_{t}\leq 0,
\end{gather*}
for some constant $C>0$. Moreover, $|u_{t}|$ is integrable in any bounded domain of $\seta$. 
\end{coro} 

\section{Variational inequality, and the boundary cases}
\subsection{Degenerate variational inequality}
\noindent
Applying the dynamic programming principle, we see that the value function satisfies the variational inequality 
\begin{align}\label{pu}
\begin{cases}
\min\big\{{-}\BL u,\; u-(x-K)^+\big\}=0, \quad (x, \pi, t)\in\seta;\\
u(x,\pi, T)=(x-K)^+,
\end{cases}
\end{align}
where
\begin{align*}
\BL u&:=u_t+\tfrac{1}{2} x^2u_{xx}+\tfrac{1}{2} \Delta^2\pi^2(1-\pi)^2u_{\pi\pi}+\Delta\pi(1-\pi)x u_{x\pi}\\
&\quad\;+( \Delta\pi+b-\gamma)x u_x+(\gamma-r)u.
\end{align*}
We will thoroughly study the free boundary of this variational inequality in the subsequent sections.
\par
The operator $\BL$ in \eqref{pu} is a parabolic operator, which is degenerate in the entire domain $\seta$. The regularity of $u(x,\pi,t)$ is not good enough to let \eqref{pu} hold almost everywhere. It holds only in weak sense, or in viscosity sense. The definition of weak solution can be found in \cite{F82} and \cite{CYW11}, and the definition of viscosity solution can be found in \cite{YZ99}.
\par
The usual way to get a weak solution for \eqref{pu} is regularization: that is, to add a term $\varepsilon u_{\pi\pi}$ to $\BL u$ for $\varepsilon>0$. One can first show there is a ``good'' $u_\varepsilon$ that solves
\begin{align}\label{puep}
\begin{cases}
\min\big\{{-}(\BL+\varepsilon\p_{\pi\pi}) u_\varepsilon,\; u_\varepsilon-(x-K)^+\big\}=0, \quad (x, \pi, t)\in\seta;\\
u_\varepsilon(x,\pi, T)=(x-K)^+,
\end{cases}
\end{align}
where suitable boundary conditions should be put on the boundaries $\pi=0, 1$. Then, to show the limit $\lim_{\varepsilon\to 0+}u_{\varepsilon}$ solves \eqref{pu} in the weak sense. In doing so, the difficulty is to get estimates, just like in lemmas 3.2-3.6 above, for $u_\varepsilon$ uniformly with respect to $\varepsilon$.
This is a long way (see \cite{CYW11}, where the regularity of the value function is the same as in the lemmas 3.2-3.6 in this paper).

\begin{remark}[On boundary conditions] 
Since the operator $\BL$ in \eqref{pu} is degenerate, according to the Fichera Theorem \cite{OR73}, boundary conditions for \eqref{pu} must not be put on $x=0$, $\pi=0$, or $\pi=1$. These boundary cases would not happen if they are initially not so. A discussion can be found in \cite{DZZ10}.
On the other hand, we can determine the boundary values on a priority basis.
\begin{itemize}
\item Let $x=0$ in \eqref{pu}, then $u(0,\pi,t)$ is the solution of the system
\begin{align*}
\begin{cases}
\min\big\{{-}u_t(0,\pi,t)-\tfrac{1}{2} \Delta^2\pi^2(1-\pi)^2u_{\pi\pi}(0,\pi,t)+ru(0,\pi,t),\; u(0,\pi,t)\big\}=0,\\
\hfill (\pi,t)\in (0,1)\times[0,T);\\
u(0,\pi, T)=0.
\end{cases}
\end{align*}
Obviously, $u(0,\pi,t)\equiv 0$ is the unique solution for this problem.

\item
Let $\pi=0$ in \eqref{pu}, then $u(x,0,t)$ is the solution for the system
\begin{align}\label{pi=0}
\begin{cases}
\min\big\{{-}u_t(x,0,t)-\tfrac{1}{2} x^2u_{xx}(x,0,t)-( b-\gamma)x u_{x}(x,0,t)+ru(x,0,t),\\
\hfill u(x,0,t)-(x-K)^+\big\}=0, \quad (x,t)\in (0,+\infty)\times [0,T);\\
u(x,0, T)=(x-K)^+.
\end{cases}
\end{align}

\item
Let $\pi=1$ in \eqref{pu}, then $u(x,1,t)$ is the solution for the system
\begin{align}\label{pi=1}
\begin{cases}
\min\big\{{-}u_t(x,1,t)-\tfrac{1}{2} x^2u_{xx}(x,1,t)-(a-\gamma)x u_{x}(x,1,t)+ru(x,1,t),\\
\hfill u(x,1,t)-(x-K)^+\big\}=0, \quad (x,t)\in (0,+\infty)\times [0,T);\\
u(x,1, T)=(x-K)^+.
\end{cases}
\end{align}
\end{itemize}
The problems \eqref{pi=0} and \eqref{pi=1} are free of drift uncertainty, and they have been thoroughly studied by Dai and Xu \cite{DX11} under the assumptions of $b\leq r$ and $a\leq r$, respectively.
\end{remark}

\subsection{The boundary cases: $\pi=0,1$}
\noindent
Before studying \eqref{pu}, it is necessary to know the situation on the boundaries $\pi=0,1$.
These two cases are similar: one only needs to interchange $a$ and $b$ in any subsequent analysis in this section. Thus, we assume $\pi=0$ for the rest of this section. As a consequence, there is no drift uncertainty and the model reduces to the classical Black-Scholes one.
\par
Let $\pi=0$ in \eqref{pu}, then $u_0(x,t):=u(x,0,t)$ is the solution of 
\begin{align}\label{u0}
\begin{cases}
\min\Big\{{-}\BL^0 u_0,\quad u_0-(x-K)^+\Big\}=0, & (x,t)\in \seta_{0};\\
u_0(x, T)=(x-K)^+,
\end{cases}
\end{align}
where
\begin{align*}
\BL^0 u &:=u_t+\tfrac{1}{2} x^2u_{xx}+(b-\gamma)x u_x+(\gamma-r)u,\\
\seta_{0}&:=(0,+\infty)\times [0,T).\nn
\end{align*}

\begin{remark}
According to the calculation in \eqref{no} below, there is no free boundary if $b\geq \gamma$.
\end{remark}
As $\gamma>r$, we only need to consider the following two cases: $\gamma>r> b $, and $\gamma> b \geq r$.
\par
In \cite{DX11}, the problem \eqref{u0} was thoroughly studied in the case of $r>b $. They have especially shown that there exists a non-increasing redeeming boundary $X_0(t)$ with the terminal value
\begin{align*}
X_0(T):=\lim_{t\to T}X_{0}(t)= \max\Big\{K,\;\frac{r-\gamma}{r-b}K\Big\}.
\end{align*}
\par
Now, we deal with the case of $\gamma> b \geq r$.
Define two free boundaries for $t\in[0,T)$, 
\begin{align*}
X_1(t) &:=\inf\{x>0\mid u_0(x,t)=(x-K)^+\}, \\
X_2(t) &:=\sup\{x>0\mid u_0(x,t)=(x-K)^+\},
\end{align*}
with the convention that $X_{1}(T):=\lim_{t\to T}X_1(t)$, $X_{2}(T):=\lim_{t\to T}X_2(t)$, and $\sup \emptyset=+\infty$. We say that $X_{1}(t)$ disappears if the set $\{x>0\mid u_0(x,t)=(x-K)^+\}$ is empty; $X_{2}(t)$ disappears if it is $+\infty$. Clearly, we have 
\begin{align*}
\{(x,t)\in\seta_{0}\mid u_{0}=(x-K)^+\}=\{(x,t)\in\seta_{0}\mid X_1(t)\leq x \leq X_2(t)\},
\end{align*}
by the continuity of the value function. 
\begin{lemma} If $\gamma> b > r$, then $X_1(t)$ is strictly decreasing and $X_2(t)$ is strictly increasing with the terminal values $X_1(T)=K$ and $X_2(T)=\tfrac{\gamma-r}{b-r}K$. 
\begin{figure}[H]
\begin{center}
\begin{picture}(350,170)
\thinlines
\put(80,40){\vector(1,0){190}}
\put(100,20){\vector(0,1){135}}
\put(100,120){\line(1,0){150}}
\qbezier(130,120)(155,100)(160,40)
\qbezier(220,120)(185,100)(180,40)
\put(258,28){$x$}
\put(92,144){$t$}
\put(88,117){$T$}
\put(125,127){$K$}
\put(126,117){$\bullet$}

\put(120,63){$X_1(t)$}
\put(135,73){\vector(1,1){10}}
\put(190,63){$X_2(t)$}
\put(203,73){\vector(-1,1){10}}
\put(205,129){$\frac{\gamma-r}{b-r}K$}
\put(218,117){$\bullet$}

\put(115,95){$\setc$}
\put(225,95){$\setc$}
\put(165,95){$\setr$}
\end{picture} \vspace{-15pt}
\caption{$\pi=0$ and $\gamma> b > r$.}
\label{fig:1}
\end{center} \vspace{-15pt}
\end{figure}
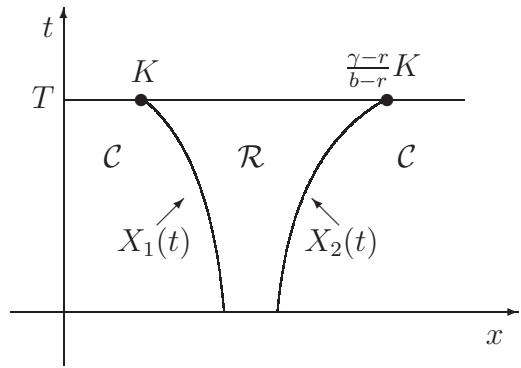
\end{lemma}

\begin{proof}
If
\begin{align*}
-\BL^0 (x-K)&=(r-b)x+(\gamma-r)K\geq 0,
\end{align*}
then $x\leq \tfrac{\gamma-r}{b-r}K$, so we have
\begin{align*}
\{(x,t)\in\seta_{0}\mid u_{0}=(x-K)^+\}\subseteq\{(x,t)\in\seta_{0}\mid K\leq x\leq \tfrac{\gamma-r}{b-r}K\}.
\end{align*}

Since $u_0$ is non-increasing in $t$, we see that $X_{1}(t)$ is non-increasing and $X_{2}(t)$ is non-decreasing. Moreover, we can prove that, by the Hopf lemma in PDE, any one of them, if they exist, is strictly monotonic. 
\par
Now, let us prove $X_1(T)=K$. First, we have $X_1(T)\geq K$. If $X_1(T)>K$, then for sufficient small $\varepsilon>0$
and $(x,t)\in (K,X_1(T))\times (T-\varepsilon,T)$,
\begin{align*}
{-}\BL^0 u_0=0,\quad u_0(x, T)=x-K,
\end{align*}
hence
\begin{align*}
\p_t u_0(x,t)=-(b-r)x+(\gamma-r)K,
\end{align*}
which is positive at $(x,t)=(K,T)$, contradicting $\p_t u_0\leq 0$. Thus, $X_1(T)=K$. In the same way, $X_2(T)=\frac{\gamma-r}{b-r}K$
can be proved.
\end{proof}

\begin{prop}\label{gamm>b-delta>r}
In the case of $\gamma> b > r$, we have
\begin{itemize}
\item If $\gamma \geq b+\tfrac{1}{2}+\sqrt{2b-2r}$,
then the two redeeming boundaries $X_1(t)$ and $X_2(t)$ exist and 
\begin{align}\label{freeboundarybounds}
X_{1}(t)<\left(1+\tfrac{1}{\sqrt{2b-2r}}\right)K<X_{2}(t)
\end{align} 
for all $0\leq t\leq T$. 
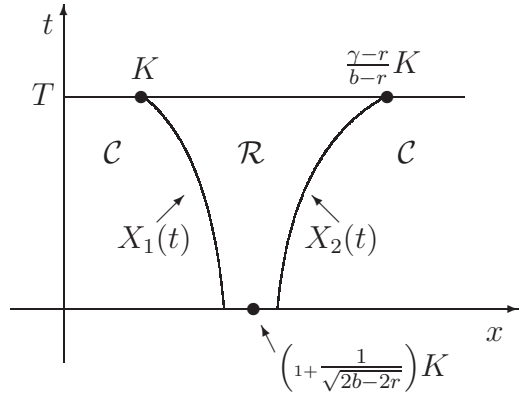
\begin{figure}[H]
\begin{center}
\begin{picture}(350,170)
\thinlines
\put(80,40){\vector(1,0){190}}
\put(100,20){\vector(0,1){135}}
\put(100,120){\line(1,0){150}}
\qbezier(130,120)(155,100)(160,40)
\qbezier(220,120)(185,100)(180,40)
\put(258,28){$x$}

\put(180,15){${\scriptscriptstyle{\big(1+\tfrac{1}{\sqrt{2b-2r}}\big)}}K$}
\put(179,24){\vector(-1,2){5}}
\put(168,37){$\bullet$}

\put(92,144){$t$}
\put(88,117){$T$}
\put(125,127){$K$}
\put(126,117){$\bullet$}

\put(120,63){$X_1(t)$}
\put(135,73){\vector(1,1){10}}
\put(190,63){$X_2(t)$}
\put(203,73){\vector(-1,1){10}}
\put(205,129){$\frac{\gamma-r}{b-r}K$}
\put(218,117){$\bullet$}

\put(115,95){$\setc$}
\put(225,95){$\setc$}
\put(165,95){$\setr$}
\end{picture} \vspace{-5pt}
\caption{$\pi=0$, $ \gamma \geq b+\tfrac{1}{2}+\sqrt{2b-2r}$ and $b > r$.}
\label{fig:1-1}
\end{center} \vspace{-15pt}
\end{figure}

\item If $\gamma< b+\tfrac{1}{2}+\sqrt{2b-2r}$,
then there exists $\ell<T$ such that 
the two redeeming boundaries $X_1(t)$ and $X_2(t)$ exist for $T-\ell\leq t\leq T$, intersect at $T-\ell$; and both of them disappear for $t<T-\ell$.\footnote{Here, $\ell$ does not depend on $T$. Further, an upper bound for $\ell$ can be explicitly given as shown in the proof.}

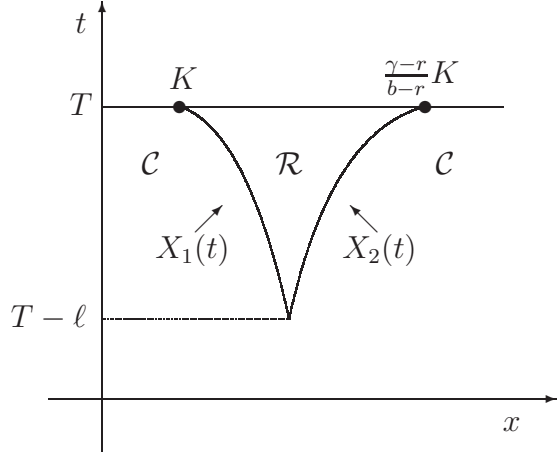
\begin{figure}[H]
\begin{center}
\begin{picture}(350,170)
\thinlines
\put(80,10){\vector(1,0){190}}
\put(100,-10){\vector(0,1){170}}
\put(100,120){\line(1,0){150}}
\qbezier(130,120)(155,110)(170,40)
\qbezier(220,120)(185,110)(170,40)
\put(250,-2){$x$}
\put(90,148){$t$}
\put(88,117){$T$}
\put(66,37){$T-\ell$}
\multiput(100,40)(1,0){70} {\line(1,0){0.4}}

\put(125,127){$K$}
\put(126,117){$\bullet$}

\put(120,63){$X_1(t)$}
\put(135,73){\vector(1,1){10}}
\put(190,63){$X_2(t)$}
\put(203,73){\vector(-1,1){10}}
\put(205,129){$\frac{\gamma-r}{b-r}K$}
\put(218,117){$\bullet$}

\put(115,95){$\setc$}
\put(225,95){$\setc$}
\put(165,95){$\setr$}
\end{picture}\vspace{10pt}
\caption{$\pi=0$ and $r< b < \gamma< b+\tfrac{1}{2}+\sqrt{2b-2r}$.}
\label{fig:2}
\end{center}\vspace{-24pt}
\end{figure}

\end{itemize}
\end{prop}
We would like to present a financial interpretation of this result before presenting its proof. 
\par
In the first case, the loan rate is relatively high. The borrower should redeem the stock when its price increases to the redeeming boundary $X_1(t)$ because the loan rate is too high to wait. If the stock price is higher than $X_2(t)$, it seems that one should wait for its price to come down to $X_2(t)$ to redeem the stock. However, no client shall sign such a stock loan with the bank at the very beginning, because the principal is too little (or the loan rate is too high) when compared to the value of the collateral (namely, the stock). 
\par
In the second case, the loan rate is too low, such that the stock loan can be almost regarded as an American call option (which shall not be exercised before maturity, in keeping with Merton's theorem). Hence, the borrower of the stock loan should wait if maturity is very far. 
\begin{proof}
Set $\gamma_{0}= b+\tfrac{1}{2}+\sqrt{2b-2r}$.
\begin{itemize}
\item Suppose $\gamma\geq \gamma_{0}$. Set $\lambda_{0}=1+\sqrt{2b-2r}$, $x_{0}=\tfrac{\lambda_{0}}{\lambda_{0}-1}K$ and define 
\begin{align*} 
v(x)&=\left(\tfrac x{\lambda_{0}}\right)^{\lambda_{0}}\left(\tfrac{\lambda_{0}-1}{K}\right)^{\lambda_{0}-1},\quad \quad x>0.
\end{align*}
From $v(x_{0})=\tfrac{1}{\lambda_{0}-1}K=x_{0}-K$, $v'(x_{0})=1$, and the fact that $v$ is strictly convex as $\lambda_{0}>1$, we conclude 
\begin{align*} 
v(x) 
=(x-K)^+,&\quad \text{if $x=x_{0}$;}\\
v(x)>(x-K)^+,&\quad \text{if $x\neq x_{0}$.}
\end{align*}
Moreover, 
\begin{align*} 
-\BL^0 v&=-v_t-\tfrac{1}{2} x^2v_{xx}-(b-\gamma)x v_x-(\gamma-r)v\\
&=(-\tfrac{1}{2} \lambda_{0}(\lambda_{0}-1)-(b-\gamma)\lambda_{0}-(\gamma-r))v\\
&=(-\tfrac{1}{2} \lambda_{0}(\lambda_{0}-1)-b\lambda_{0}+\gamma(\lambda_{0}-1)+r)v\\
&\geq(-\tfrac{1}{2} \lambda_{0}(\lambda_{0}-1)-b\lambda_{0}+\gamma_{0}(\lambda_{0}-1)+r)v\\
&=(-\tfrac{1}{2} \lambda_{0}(\lambda_{0}-1)-b\lambda_{0}+(\lambda_{0}+b-\tfrac{1}{2})(\lambda_{0}-1)+r)v\\
&=(\tfrac{1}{2} (\lambda_{0}-1)^{2}-b+r)v\\
&=0.
\end{align*}
Therefore, $v$ is a super solution of \eqref{u0}, and hence
\[(x-K)^+\leq u_{0}(x,t)\leq v(x), \quad (x,t)\in\seta,\] 
from which we get $u_0(x_{0},t)=(x_{0}-K)^{+}$ for $0\leq t\leq T$. Notice \[x_{0}=\tfrac{\lambda_{0}}{\lambda_{0}-1}K=\left(1+\tfrac{1}{\sqrt{2b-2r}}\right)K,\] so \eqref{freeboundarybounds} follows from the fact that $X_1(t)$ and $X_2(t)$ are both strictly monotonic. 

\item 
Now, suppose $\gamma<\gamma_{0}$. 
We note that by \eqref{umonoexp}
\begin{align*} 
u_{0}(x,t)
&=\sup_{\tau\in\sett_{[0,T-t]}}\;\BE{e^{(\gamma-r)\tau}\left(xe^{\int_{0}^{\tau}(\Delta\pi_{\nu}+b-\gamma-\frac{1}{2})\dnu+\overline{W}_{\tau} }-K\right)^{+}\bigg| \pi_{0}=0 }\\
&=\sup_{\tau\in\sett_{[0,T-t]}}\;\BE{e^{(\gamma-r)\tau}\left(xe^{(b-\gamma-\frac{1}{2})\tau+\overline{W}_{\tau} }-K\right)^{+}}\\
&\geq \BE{e^{(\gamma-r)(T-t)}\left(xe^{(b-\gamma-\frac{1}{2})(T-t)+\overline{W}_{T-t} }-K\right)^{+}}\\
&=xe^{(b-r)(T-t)} N\left(\tfrac{\log x-\log K-(\gamma-b-\frac{1}{2})(T-t)}{\sqrt{T-t}}\right)\\ 
&\quad -Ke^{(\gamma-r)(T-t)}N\left(\tfrac{\log x-\log K-(\gamma-b+\frac{1}{2})(T-t)}{\sqrt{T-t}}\right)=:g(x,t),
\end{align*} 
where $N$ stands for the standard normal distribution. 
We note that for $x\geq K$,
\begin{align*} 
\p_{x}g(x,t)&=e^{(b-r)(T-t)} N\left(\tfrac{\log x-\log K-(\gamma-b-\frac{1}{2})(T-t)}{\sqrt{T-t}}\right)\\
&\geq e^{(b-r)(T-t)} N\left(-(\gamma-b-\tfrac{1}{2})\sqrt{T-t}\right).
\end{align*} 
If we can prove that the right hand side in the above equation is $\geq 1$ for sufficiently large $T-t$, then $g(x,t)>g(x,t)-g(K,t)\geq x-K$ for all $x\geq K$ as $g(K,t)>0$ by definition. Consequently, we deduce $u_{0}(x,t)>(x-K)^{+}$ for sufficiently large $T-t$ and the claim follows.
\begin{itemize}
\item If $\gamma\leq b+\tfrac{1}{2}$, then, $\p_{x}g(x,t)\geq e^{(b-r)(T-t)}N(0)\geq 1$ for sufficiently large $T-t$. 
\item If $b+\tfrac{1}{2}<\gamma<\gamma_{0}$, using the well-known inequality 
\[N(-x)\geq \frac{1}{\sqrt{2\pi}}\frac{x}{(x^{2}+1)}e^{-\frac{1}{2}x^{2}},\quad x\geq 0,\] 
we have for sufficiently large $T-t$, 
\begin{align*} 
\p_{x}g(x,t)&\geq e^{(b-r)(T-t)} N\left(-(\gamma-b-\tfrac{1}{2})\sqrt{T-t}\right)\\
&\geq \frac{1}{\sqrt{2\pi}}\frac{(\gamma-b-\tfrac{1}{2})\sqrt{T-t}}{(\gamma-b-\tfrac{1}{2})^{2}(T-t)+1} e^{\left(b-r-\tfrac{1}{2}(\gamma-b-\tfrac{1}{2})^{2}\right)(T-t)} \geq 1,
\end{align*} 
in view of $b-r-\tfrac{1}{2}(\gamma-b-\tfrac{1}{2})^{2}>b-r-\tfrac{1}{2}(\gamma_{0}-b-\tfrac{1}{2})^{2}=0$.
\end{itemize}
\end{itemize}
\end{proof}

\begin{prop}\label{pi=0;r=b-delta}
In the case of $\gamma> b=r$, there is only one non-increasing redeeming boundary $X_{1}(t)$ with the terminal value $X_{1}(T)=K$. Moreover, the redeeming boundary satisfies the following properties.
\begin{itemize}
\item If $\gamma>r+\tfrac{1}{2}$, then $X_{1}(t)$ exists and $X_{1}(t)\leq \tfrac{2(\gamma-r)}{2(\gamma-r)-1}K$ for $0\leq t\leq T$.\footnote{In this case, one can prove $\lim\limits_{T-t\to+\infty}X_{1}(T-t)=\tfrac{2(\gamma-r)}{2(\gamma-r)-1}K$. Due to the space limitation, we leave the proof to the interested readers.}
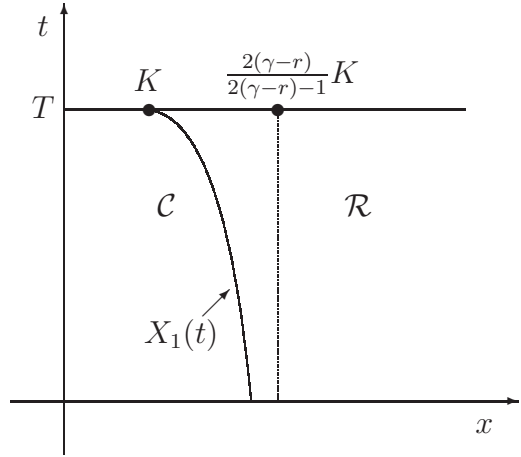
\begin{figure}[H]
\begin{center}
\begin{picture}(350,170)
\thinlines
\put(80,10){\vector(1,0){190}}
\put(100,-10){\vector(0,1){170}}
\put(100,120){\line(1,0){150}}
\qbezier(130,120)(160,120)(170,10)
\put(254,-2){$x$}
\put(90,148){$t$}
\put(88,117){$T$}

\put(126,127){$K$}
\put(129,117){$\bullet$}

\multiput(180,120)(0,-1){110} {\line(0,-1){0.4}}
\put(177,117){$\bullet$}
\put(160,130){$\tfrac{2(\gamma-r)}{2(\gamma-r)-1}K$}

\put(130,32){$X_{1}(t)$}
\put(152,42){\vector(1,1){10}}
\put(135,80){$\setc$}
\put(205,80){$\setr$}
\end{picture}\vspace{10pt}
\caption{$\pi=0$, $b=r$ and $\gamma> r+\tfrac{1}{2}$.}
\label{fig:10-1}
\end{center}\vspace{-14pt}
\end{figure}
\item If $\gamma\leq r+\tfrac{1}{2}$, then $\lim\limits_{T-t\to+\infty}X_{1}(T-t)=+\infty$.\footnote{We suspect that $X_{1}$ disappears for some finite time $t$ in this case, but cannot prove it. }
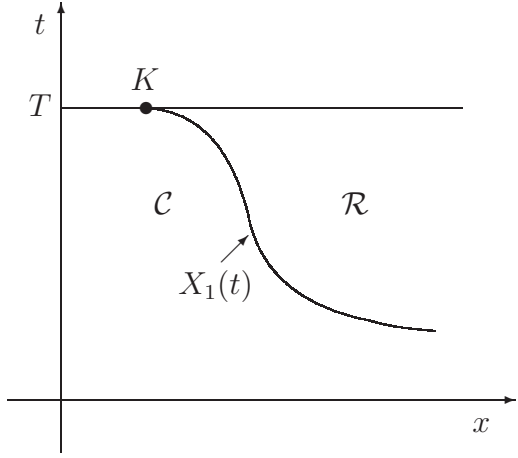
\begin{figure}[H]
\begin{center}
\begin{picture}(350,170)
\thinlines
\put(80,10){\vector(1,0){190}}
\put(100,-10){\vector(0,1){170}}
\put(100,120){\line(1,0){150}}
\qbezier(130,120)(160,120)(170,80)
\qbezier(170,80)(176,48)(215,40)
\qbezier(215,40)(225,37)(240,36)
\put(254,-2){$x$}
\put(90,148){$t$}
\put(88,117){$T$} 

\put(126,127){$K$}
\put(129,117){$\bullet$}

\put(144,50){$X_{1}(t)$}
\put(160,62){\vector(1,1){10}}
\put(135,80){$\setc$}
\put(205,80){$\setr$}
\end{picture}\vspace{10pt}
\caption{$\pi=0$ and $b=r < \gamma\leq r+\tfrac{1}{2}$.}
\label{fig:10}
\end{center}\vspace{-14pt}
\end{figure}
\end{itemize}
\end{prop}

\begin{proof}
Assume $\gamma>b=r$. By \citelem{lem:uxbound}, we have 
\[0\leq u_{0}(x,t)-u_{0}(y,t)\leq x-y,\quad x\geq y\geq 0,\]
that is 
\[ -K\leq u_{0}(x,t)-x\leq u_{0}(y,t)-y.\]
Therefore, if $u_{0}(y,t)=y-K$ for some $y$, then $u_{0}(x,t)=x-K$ for any $x\geq y$. Hence, the free boundary $X_{2}(t)$ disappears.

\begin{itemize}
\item Suppose $\gamma > r+\tfrac{1}{2}$. Let $\lambda_1=2(\gamma-r)>1$, $x_{1}=\frac{\lambda_1}{\lambda_1-1}K$ and
\begin{align*}
w(x)=\begin{cases}
\left(\frac{x}{\lambda_1}\right)^{\lambda_1}\left(\frac{\lambda_1-1}{K}\right)^{\lambda_1-1},&\quad 0<x<x_{1};\\
(x-K)^{+},&\quad x\geq x_{1}.
\end{cases}
\end{align*}
It can be seen that $w(x)$ is a super solution of \eqref{u0}, hence
\[u_{0}(x,t)\leq w(x)=(x-K)^{+}, \quad x\geq x_{1}=\tfrac{2(\gamma-r)}{2(\gamma-r)-1}K.\]
The claim follows immediately.

\item Suppose $\gamma \leq r+\tfrac{1}{2}$. 
Since $x$ is a super solution of \eqref{u0} with $b=r$, we have $u_0(x,t)\leq x$. Set
\[u_0^\infty(x)=\lim\limits_{T-t\to+\infty}u_0(x,t)\leq x.\]
Then, it solves the corresponding stationary problem of \eqref{u0}
\begin{align*}
\min \Big\{-\tfrac{1}{2} x^2\p_{xx}u_0^\infty(x)+(\gamma-r) (x\p_{x}u_0^\infty(x)-u_0^\infty(x)),\quad u_0^\infty(x)-\left(x-K \right)^+\Big\}=0, \;\; x>0.
\end{align*}
Let \[x_{0}=\lim\limits_{T-t\to+\infty}X_{1}(T-t)\in[K,+\infty].\] 
Then $u_0^\infty$ and $x_0$ satisfy
\begin{align}\label{4.10}
\begin{cases}
-\tfrac{1}{2} x^2\p_{xx}u_0^\infty(x)+(\gamma-r) (x\p_{x}u_0^\infty(x)-u_0^\infty(x))=0, & 0<x<x_0;\\
\left(x-K\right)^+\leq u_0^\infty(x)\leq x, & x>0.
\end{cases}
\end{align}
Moreover, 
\begin{align}\label{4.11}
\text{if $x_{0}<+\infty$, then $u_0^\infty(x_0)=x_{0}-K$, and $\p_x u_0^\infty(x_0)=1$. }
\end{align} 
The general solution of the equation \eqref{4.10} is 
\[u_0^\infty(x)=
\begin{cases}
C_1 x\ln x +C_2x, & \quad \text{if $\gamma= r+\tfrac{1}{2}$};\\
C_1 x^{2(\gamma-r)} +C_2x,& \quad \text{if $\gamma< r+\tfrac{1}{2}$};\\
\end{cases}
\quad 0<x<x_{0}.
\]
Letting $x$ go to 0 in the condition $\left(1-\tfrac{K}{x} \right)^+\leq \frac{u_0^\infty(x)}{x}\leq 1$, we obtain $C_1=0$. 
Hence, $u_0^\infty(x)=C_{2}x$ for $0<x<x_{0}$. Consequently, the condition \eqref{4.11} cannot be satisfied for any $x_{0}<\infty$. So $x_{0}=+\infty$, and the claim follows. 
\end{itemize}
The proof is complete.
\end{proof}

\section{General case: $0<\pi<1$}
\noindent
Now, we go back to the general case $0<\pi<1$.
\par
We will discuss the behaviors of the redeeming boundaries. From \citelem{umono} we see that
\begin{align}
\p_t\left(u(x,\pi,t)-(x-K)^+\right)& \leq 0, \quad (x,\pi, t)\in\seta;\label{ut}\\
\p_\pi\left(u(x,\pi,t)-(x-K)^+\right)&\geq 0, \quad (x,\pi, t)\in\seta.\label{upi}
\end{align}
Moreover, \eqref{uxbound2} reveals
\begin{align}\label{ux}
\p_x\left(u(x,\pi,t)-(x-K)^+\right) &\leq 0, \quad \text{if }r\geq a,\; x\geq K, \; (x,\pi, t)\in\seta.
\end{align}
\par
Define, for each fixed $(\pi,t)\in(0,1)\times[0,T)$,
\begin{align}
X_1(\pi,t)&:=\min\{x\mid u(x,\pi,t)=(x-K)^+,\; (x,\pi,t)\in\seta\}, \label{x1def} \\
X_2(\pi,t)&:=\max\{x\mid u(x,\pi,t)=(x-K)^+,\; (x,\pi,t)\in\seta\}.\label{x2def}
\end{align}
Then
\begin{lemma}\label{regionbyx1x2}
We have
\begin{align}
\setc&=\{(x,\pi,t)\in\seta \mid x<X_1(\pi,t)\;\text{or}\; x>X_2(\pi,t)\},\\
\setr&=\{(x,\pi,t)\in\seta \mid X_1(\pi,t)\leq x \leq X_2(\pi,t)\}.\label{y1y2}
\end{align}
Moreover, $X_1(\pi,t)$ is non-increasing in $t$ and non-decreasing in $\pi$, while $X_2(\pi,t)$
is non-decreasing in $t$ and non-increasing in $\pi$, if they exist; $X_1(\pi,t)>K$.
\end{lemma}
\begin{proof}
For each fixed $(\pi,t)\in(0,1)\times[0,T)$, by \citelem{umono}, the function $x\mapsto u(x,\pi,t)$ is convex; hence, $u\leq x-K$ must be an interval of $x$, which clearly implies the expressions $\setc$ and $\setr$. The monotonicity of $X_1(\pi,t)$ and $X_2(\pi,t)$ is a simple consequence of \eqref{ut} and \eqref{upi}.
\end{proof}
Similarly, for each fixed $(x,t)\in(0,\infty)\times[0,T)$, define
\begin{align}\label{pidef}
\Pi(x,t)&:=\max\{\pi\mid u(x,\pi,t)=(x-K)^+,\; (x,\pi,t)\in\seta\},
\end{align}
with $\Pi(x,t)=0$ when the set on the right hand side is empty. 
Then, by \eqref{ut} and \eqref{upi}, 
\begin{lemma}\label{regionbypi}
We have
\begin{align*}
\setc &=\{(x,\pi,t)\in\seta\mid \pi>\Pi(x,t)\},\\
\setr &=\{(x,\pi,t)\in\seta\mid \pi\leq \Pi(x,t)\}.
\end{align*}
Moreover, $\Pi(x,t)$ is non-decreasing in $t$.
\end{lemma}
\par
The target of this paper is reduced to studying the properties of these redeeming boundaries $X_1(\pi,t)$, $X_2(\pi,t)$, and $\Pi(x,t)$.
\par
Since $\gamma>r$, we only need to consider the following cases:
\begin{itemize}
\item[] Case 0: $b\geq \gamma>r$
\item[] Case 1: $r\geq a $
\item[] Case 2: $\gamma>b>r$ 
\item[] Case 3: $\gamma>r>b$ and $a>r$
\item[] Case 4: $\gamma>r=b$
\end{itemize}
We start with the simplest Case 0.

\subsection{Case 0: $b\geq \gamma>r$}
\noindent
If $b\geq \gamma>r$, a simple calculation shows for $x\geq K$
\begin{align}\label{no}
-\BL (x-K)&=-\BL x+\BL K\nn\\
&=(r-(\Delta\pi+b ))x+(\gamma-r)K\nn\\
&< (r-b)x+(\gamma-r)K\nn\\
&\leq (r-b)K+(\gamma-r)K\nn\\
&=(\gamma-b)K\nn\\
&\leq 0,
\end{align}
then we see that $\setr=\emptyset$ and $\setc=\seta$.
\begin{thm}
Suppose $b\geq \gamma>r$, then the optimal redeeming time for the problem \eqref{main00} is
\[\tau^{*}=T.\] 
Economically speaking, it is never optimal to redeeming the stock early because the loan rate is very low.
\end{thm}
\subsection{Case 1: $r\geq a$}
\noindent
Recalling our definition of the continuation region \eqref{C} and the redeeming region \eqref{R},
we see from \eqref{ux} that the redeeming boundary $X_{2}$ disappears if $r\geq a$.
\begin{thm}
If $r\geq a$, then
\begin{align*}
\setc=\{(x,\pi,t)\in\seta\mid x<X_{1}(\pi,t)\},\\
\setr=\{(x,\pi,t)\in\seta\mid x\geq X_{1}(\pi,t)\}.
\end{align*}
Economically speaking, one should redeem the stock if the stock price is high enough when the discounting rate is too when high compared to the return rate of the stock.
\end{thm} \par
Now, we can summarize the results of the problem \eqref{pu}.
\begin{thm} Suppose $r\geq a $, then the redeeming boundary can be expressed as $x=X_{1}(\pi, t)$, where $X_{1}(\pi,t)$ is non-decreasing in $\pi$, non-increasing in $t$, and $X_{1}(\pi,T)=K$ for any $0\leq\pi\leq 1$.
The optimal redeeming time for the problem \eqref{main00} is given by
\[\tau^{*}=\min\{\nu\in[t,T]\mid X_{\nu}\geqslant X_{1}(\pi_{\nu},\nu)\}.\]
\end{thm}

\subsection{Case 2: $\gamma>b>r$}
\noindent
In this case, the inequality in \eqref{ux} may not hold.
\begin{thm} Suppose $\gamma>b>r$, then
\begin{align}\label{41}
X_1(\pi,t)> K,\quad X_2(\pi,t)< \tfrac{\gamma-r}{b-r}K.
\end{align}
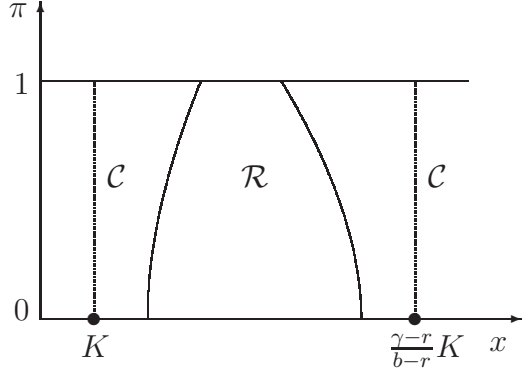
\begin{figure}[H]
\begin{center}
\begin{picture}(160,150)
\put(0,110){\line(1,0){160}}
\put(0,20){\vector(1,0){180}}
\put(0,20){\vector(0,1){120}}
\qbezier(40,20)(40,60)(60,110)
\qbezier(120,20)(120,60)(90,110)

\multiput(20,110)(0,-1){90}{\line(1,0){0.5}}
\put(15,6){$K$}
\put(17,17){$\bullet$}

\multiput(140,110)(0,-1){90}{\line(1,0){0.5}}
\put(130,6){$\tfrac{\gamma-r}{b-r}K$}
\put(137,17){$\bullet$}

\put(168,8){$x$}
\put(-10,20){$0$}
\put(-10,105){$1$}
\put(-12,134){$\pi$}
\put(75,70){$\setr $}
\put(145,70){$\setc $}
\put(25,70){$\setc $}
\end{picture}\vspace{-10pt}
\caption{In the plan $t$-section, $\gamma >b>r$.}
\label{fig:5}
\end{center}
\end{figure}
\end{thm}
\begin{proof} Recall \eqref{no},
\begin{align}\label{53}
-\BL (x-K)&=(r-(\Delta\pi+b ))x+(\gamma-r)K,
\end{align}
and note $r-(\Delta\pi+b )< r-b<0$,
so if $-\BL (x-K)\geq 0$, then
\[ x\leq \tfrac{\gamma-r}{ \Delta\pi+b-r}K< \tfrac{\gamma-r}{b-r}K.\]
Thus, \eqref{41} follows from this and \eqref{regions}.
\end{proof}

\begin{remark}\label{remarkcase2}
If $a >\gamma>b $, then $0<\tfrac{\gamma-b}{\Delta}<1$.
Suppose $\pi\geq \tfrac{\gamma-b}{\Delta}$, then $\Delta\pi+b\geq \gamma$, hence, by \eqref{53}, for $x> K$,
\begin{align*}
-\BL (x-K)&=(r-(\Delta\pi+b ))x+(\gamma-r)K
\leq (r-\gamma)x+(\gamma-r)K< 0.
\end{align*}
Recalling \eqref{regions}, we see that $\left\{(x,\pi,t)\in\seta \mid \frac{\gamma-b}{\Delta}\leq\pi<1\right\}\subseteq\setc $ (see Figure \ref{fig:6}).
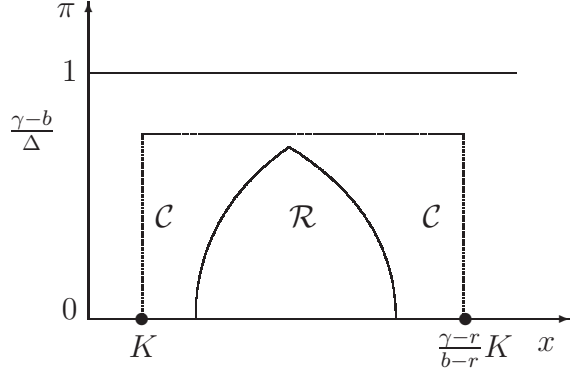
\begin{figure}[H]
\begin{center}
\begin{picture}(160,150)
\put(0,113){\line(1,0){160}}
\put(0,20){\vector(1,0){180}}
\put(0,20){\vector(0,1){120}}
\qbezier(40,20)(40,60)(75,85)
\qbezier(115,20)(115,60)(75,85)
\put(-30,88){$\tfrac{\gamma-b}{\Delta}$}
\multiput(20,90)(1,0){120}
{\line(1,0){0.5}}

\multiput(20,90)(0,-1){70}{\line(1,0){0.5}}
\put(15,6){$K$}
\put(17,17){$\bullet$}

\multiput(140,90)(0,-1){70}{\line(1,0){0.5}}
\put(130,6){$\tfrac{\gamma-r}{b-r}K$}
\put(138,17){$\bullet$}

\put(168,8){$x$}
\put(-10,20){$0$}
\put(-10,110){$1$}
\put(-12,134){$\pi$}
\put(75,55){$\setr $}
\put(125,55){$\setc $}
\put(25,55){$\setc $}
\end{picture}\vspace{-10pt}
\caption{In the plan $t$-section, $a>\gamma >b>r$.}
\label{fig:6}
\end{center}
\end{figure}
\end{remark}

\subsection{Case 3: $\gamma>r>b$ and $a>r$}
\noindent
We consider two sub-regions separately.
\begin{itemize}
\item $\tfrac{\gamma-b}{\Delta}\leq\pi<1:$ By Remark \ref{remarkcase2}, if $a>\gamma >r>b$, then \[\left\{(x,\pi,t)\in\seta \mid \tfrac{\gamma-b}{\Delta}\leq\pi<1\right\}\subseteq\setc\]
and hence no redeeming boundary exists in this region. The above relationship holds when $\gamma\geq a >r>b$ as the left set becomes the empty set.

\item $0<\pi< \frac{\gamma-b}{\Delta}$: In this region if $ \frac{(\gamma-r)K}{\Delta x}+ \frac{r-b}{\Delta}<\pi< \frac{\gamma-b}{\Delta}$, we have $(\Delta\pi+b-r)x>(\gamma-r)K$, so by \eqref{53},
\begin{align*}
-\BL (x-K)&=(r-(\Delta\pi+b ))x+(\gamma-r)K<0,
\end{align*}
hence
\[\left\{(x,\pi,t)\in\seta \mid \tfrac{(\gamma-r)K}{\Delta x}+ \tfrac{r-b}{\Delta}<\pi< \tfrac{\gamma-b}{\Delta}\right\}\subseteq\setc. \]
\end{itemize}
Recalling \citelem{regionbypi}, we summarize the result obtained thus far in
\begin{thm}
If $\gamma >r>b$ and $a>r$, then
\begin{align*}
\setc &=\{(x,\pi,t)\in\seta\mid \pi>\Pi(x,t)\},\\
\setr &=\{(x,\pi,t)\in\seta\mid \pi\leq \Pi(x,t)\},
\end{align*}
where the redeeming boundary $\Pi(x,t)$ is first non-decreasing, and then decreasing in $x$. 
The non-decreasing part is $x=X_{1}(\pi,t)$, and the non-increasing part is $x=X_{2}(\pi,t)$.
Moreover, it is universally upper-bounded by
\[\Pi(x,t)< \tfrac{\gamma-b}{\Delta}\]
when $a>\gamma$. 
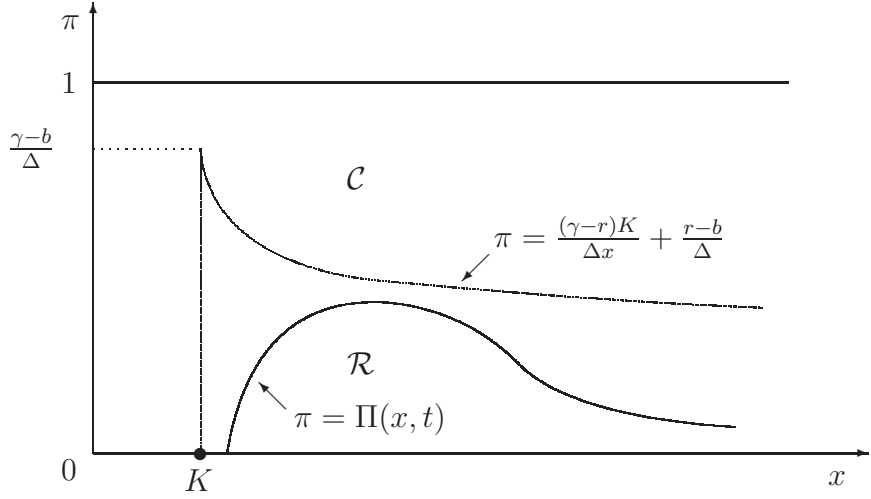
\begin{figure}[H]
\begin{center}
\begin{picture}(440,150)
\multiput(80,85)(3,0){13}{\line(1,0){0.5}}
\put(48,82){$\frac{\gamma-b}{\Delta}$}
\multiput(120,85)(0,-1){115}{\line(1,0){0.5}}
\put(114,-44){$K$}
\put(117,-33){$\bullet$}

\put(80,110){\line(1,0){260}}
\put(80,-30){\vector(1,0){290}}
\put(80,-30){\vector(0,1){170}}

\qbezier(130,-30)(140,30)(190,27) 
\qbezier(190,27)(220,24)(240,3)
\qbezier(240,3)(260,-16)(320,-20)
\qbezier[100](120,85)(124,45)(180,36)
\qbezier[150](180,36)(260,28)(330,25)

\put(230,48){$\pi=\frac{(\gamma-r)K}{\Delta x}+ \frac{r-b}{\Delta}$}
\put(228,45){\vector(-1,-1){10}}
\put(355,-42){$x$}
\put(68,-40){$0$}
\put(68,106){$1$}
\put(68,130){$\pi$}
\put(155,-20){$\pi=\Pi(x,t)$}
\put(152,-13){\vector(-1,1){10}}
\put(175,0){$\setr $}
\put(175,70){$\setc $}
\end{picture}\vspace{45pt}
\caption{In the plan $t$-section, $a>\gamma >r>b$.}
\label{fig:7}
\end{center}
\end{figure}
Economically speaking, one should not redeem the stock if the excess return rate of the stock is very likely to be $a$ (which is higher than the loan rate). 
\end{thm}

\subsection{Case 4: $\gamma>r=b$}
\noindent
In this case, we have
\begin{thm}
If $\gamma >r=b$, then
\begin{align*}
\setc &=\{(x,\pi,t)\in\seta\mid \pi>\Pi(x,t)\},\\
\setr &=\{(x,\pi,t)\in\seta\mid \pi\leq \Pi(x,t)\},
\end{align*}
where the redeeming boundary $\Pi(x,t)$ is first non-decreasing and then non-increasing in $x$ with
\[\lim_{x\to+\infty}\Pi(x,t)=0.\]
Further, $\Pi(x,t)<\frac{\gamma-r}{\Delta}$ when $a>\gamma$. 
\end{thm}
\begin{proof}
We first prove that $\Pi(x,t)$ is positive for all sufficiently large $x$.
Suppose this is not so, because $\Pi(x,t_0)$ is non-increasing for large $x$; then there is a point $ (x_0, 0, t_0)$ such that $\Pi(x,t_0)=0$ for all $x>x_0$.
From \citeprop{pi=0;r=b-delta} and the definition of $\Pi$ \eqref{pidef} we see that
\begin{align*}
u(x,0,t_0)&=x-K,\quad x> x_0; \\
u(x,\pi,t_0)&>x-K,\quad x> x_0,\;\pi>0.
\end{align*}
By $u_t\leq 0$, we have
\begin{align}
u(x,0,t)&=x-K,\quad x\geq x_0, \; t\geq t_{0};\label{u(x0t)}\\
u(x,\pi,t)&>x-K,\quad x\geq x_0,\;\pi>0,\;t \leq t_0.
\end{align}
Hence by \eqref{pu}
\begin{multline*}
u_t+\tfrac{1}{2} x^2u_{xx}+\tfrac{1}{2}
\Delta^2\pi^2(1-\pi)^2u_{\pi\pi}+\Delta\pi(1-\pi)x u_{x\pi}\\
\quad\quad+( \Delta\pi+b-\gamma)x u_x+(\gamma-r)u=0,\quad x\geq
x_0,\;\pi>0,\;t \leq t_0.
\end{multline*}
Letting $\pi$ go to 0,
\begin{align}\label{nearpi=0}
u_t+\tfrac{1}{2}x^2u_{xx}+( b-\gamma)x
u_x+(\gamma-r)u=0,\quad x\geq x_0,\;\pi=0,\;t \leq t_0.
\end{align}
By \eqref{u(x0t)}, we have for $x\geq x_0$, $\pi=0$, $t=t_0$,
\begin{align*}
u_t(x,0,t_0)\leq 0,\quad u_{xx}(x,0,t_0)=0,\quad u_{x}(x,0,t_0)=1,
\quad u(x,0,t_0)=x-K.
\end{align*}
Substituting them into the equation \eqref{nearpi=0} with $t=t_{0}$,
\[u_t(x,0,t_0)+( b-\gamma)x+(\gamma-r)(x-K)=0.\]
Applying $r=b$, we obtain $u_t(x,0,t_0)=(\gamma-r)K>0$, which is a
contradiction.
\par The non-decreasing part of $\pi=\Pi(x,t)$ corresponds to $x=X_1(\pi,t)$ in \eqref{x1def}, and the non-increasing part corresponds to $x=X_2(\pi,t)$ in \eqref{x2def}. 
\par
For $\pi> \frac{(\gamma-r)K}{\Delta x}$,
we have by \eqref{53},
\begin{align*}
-\BL (x-K)&=(r-(\Delta\pi+b ))x+(\gamma-r)K=-\Delta\pi x+(\gamma-r)K<0,
\end{align*}
hence by \eqref{regions}
\[\left\{(x,\pi,t)\in\seta \mid \pi>\tfrac{(\gamma-r)K}{\Delta x}\right\}\subseteq\setc. \]
This clearly implies $\lim_{x\to +\infty}\Pi(x,t)=0$ and $\Pi(x,t)<\frac{\gamma-r}{\Delta}$ when $a>\gamma$ as $\Pi(x,t)>K$.
\end{proof}

See Figures \ref{fig:8}, \ref{fig:9} for illustrations of our results when $a\leq \gamma$ and $a>\gamma$.
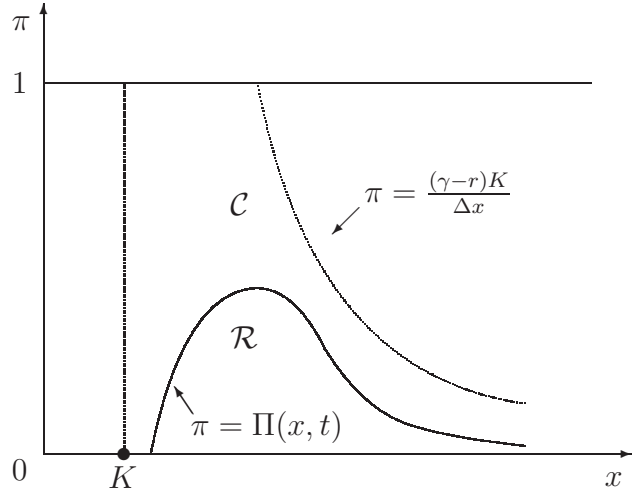
\begin{figure}[H]
\begin{center}
\begin{picture}(385,150)
\multiput(110,110)(0,-1){140}{\line(1,0){0.5}}
\put(104,-43){$K$}
\put(107,-33){$\bullet$}

\put(80,110){\line(1,0){205}}
\put(68,106){$1$}

\put(80,-30){\vector(1,0){220}}
\put(290,-42){$x$}
\put(68,-40){$0$}

\put(80,-30){\vector(0,1){170}}
\put(68,130){$\pi$}

\qbezier(120,-30)(130,20)(150,30)
\qbezier(150,30)(170,40)(185,10)
\qbezier(185,10)(200,-15)(220,-20)
\qbezier(220,-20)(233,-24)(260,-27)

\qbezier[160](160,110)(180,0)(260,-11)
\put(200,65){$\pi=\tfrac{(\gamma-r)K}{\Delta x}$}

\put(198,63){\vector(-1,-1){10}}
\put(135,-23){$\pi=\Pi(x,t)$}
\put(136,-15){\vector(-2,3){7}}
\put(150,10){$\setr $}
\put(150,60){$\setc $}
\end{picture}\vspace{45pt}
\caption{In the plan $t$-section, $\gamma\geq a >r=b$.}
\label{fig:8}
\end{center}
\end{figure}

\begin{figure}[H]
\begin{center}
\begin{picture}(385,150)
\multiput(110,85)(0,-1){115}{\line(1,0){0.5}}
\put(104,-43){$K$}
\put(107,-33){$\bullet$}

\put(80,110){\line(1,0){205}}
\put(68,106){$1$}

\put(80,-30){\vector(1,0){220}}
\put(290,-42){$x$}
\put(68,-40){$0$}

\put(80,-30){\vector(0,1){170}}
\put(68,130){$\pi$}

\multiput(80,85)(3,0){11}{\line(1,0){0.5}}
\put(48,82){$\frac{\gamma-r}{\Delta}$}

\qbezier(120,-30)(125,10)(144,20)
\qbezier(144,20)(160,30)(175,5)
\qbezier(175,5)(190,-15)(205,-20)
\qbezier(205,-20)(213,-24)(240,-27)

\qbezier[160](110,85)(120,30)(250,1)
\put(200,34){$\pi=\tfrac{(\gamma-r)K}{\Delta x}$}
\put(198,32){\vector(-1,-1){10}}

\put(135,-23){$\pi=\Pi(x,t)$}
\put(136,-15){\vector(-2,3){7}}
\put(150,0){$\setr $}
\put(150,65){$\setc $}
\end{picture}\vspace{45pt}
\caption{In the plan $t$-section, $a>\gamma >r=b$.}
\label{fig:9}
\end{center}
\end{figure}
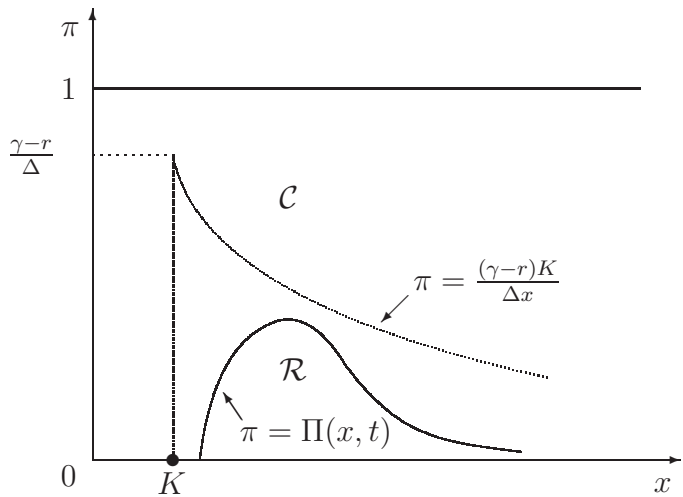

\section{Concluding remarks}
\noindent
Due to space limitations, this paper has studied only the optimal redeeming problem of the standard stock loan. Clearly, one can use our method to study stock loans with other ways of dividend distribution. Such problems, however, will involve higher dimensional processes and thus, a more involved analysis is required. We hope to investigate these areas in our future research endeavors. 
\par
Meanwhile, one can also consider the possible extensions of this study where the drift dependents on a finite state Markov chain. A similar technique should work using the Wonham filter.

\newpage



\end{document}